\documentclass{article}

\usepackage{graphicx}

\RequirePackage[OT1]{fontenc}
\RequirePackage{amsmath,amsfonts,amsthm}
\usepackage{authblk}
\newtheorem{theorem}{Theorem}

\newtheorem{lemma}{Lemma}
\newtheorem{proposition}{Proposition}

\DeclareMathOperator*{\argmax}{argmax}

\newcommand{\mJ}{{\mathcal J}}

\newcommand{\mN}{{\mathcal N}}

\newcommand{\mR}{{\mathcal R}}
\newcommand{\mS}{{\mathcal S}}

\newcommand{\mX}{{\mathcal X}}

\newcommand{\bN}{{\mathbb N}}
\newcommand{\bR}{{\mathbb R}}
\newcommand{\bZ}{{\mathbb Z}}

\newcommand{\bE}{{\mathbb E}}

\newcommand{\T}{\textsf{T}}

\newcommand{\coS}{<\!{\mathcal S}\! >}
\newcommand{\coqS}[1]{<\!{\mathcal S}\wedge Q(#1)\! >}

\begin{document}
%
%
%
%
 \title{Concave Switching in Single and Multihop Networks}
\author[1]{Neil Walton}
\affil[1]{University of Amsterdam, n.s.walton@uva.nl}
\date{}
%
%


\maketitle

\abstract{Switched queueing networks model wireless networks, input queued switches and numerous other networked communications systems. For single-hop networks, we consider a {($\alpha,g$)-switch policy} which combines the MaxWeight policies with bandwidth sharing networks -- a further well studied model of Internet congestion. We prove the maximum stability property for this class of randomized policies. Thus these policies have the same first order behavior as the MaxWeight policies. However, for multihop networks some of these generalized polices address a number of critical weakness of the MaxWeight/BackPressure policies.

For multihop networks with fixed routing, we consider the Proportional Scheduler (or (1,log)-policy). In this setting, the BackPressure policy is maximum stable, but must maintain a queue for every route-destination, which typically grows rapidly with a network's size. However, this proportionally fair policy only needs to maintain a queue for each outgoing link, which is typically bounded in number.  As is common with Internet routing, by maintaining per-link queueing each node only needs to know the next hop for each packet and not its entire route.  Further, in contrast to BackPressure, the Proportional Scheduler does not compare downstream queue lengths to determine weights, only local link information is required. This leads to greater potential for decomposed implementations of the policy.
Through a reduction argument and an entropy argument, we demonstrate that, whilst maintaining substantially less queueing overhead, the Proportional Scheduler achieves maximum throughput stability.
}

\section{Introduction}\label{sec1}
Switch networks are queueing networks were there are constraints on which queues can be served simultaneously. At each time a policy must chose a schedule that satisfies these constraints.
The MaxWeight/BackPressure scheduling policies were first introduced by Tassiulas and Ephremides  as a model of wireless communication \cite{TaEp92}. Their policy was applicable to the class of switched queueing networks both for single-hop networks -- where packets are served once before departing  -- and for multi-hop networks -- where packets are served a multiple queues before departing. For single-hop networks, their policy is called \emph{MaxWeight} and, for multihop networks, it is called \emph{BackPressure}.  The MaxWeight and BackPressure scheduling policies have the key properties of having a maximal stability region while not requiring explicit estimation of arrival rates.
Subsequently, as a model of Internet Protocol routers, McKeown et al. \cite{MMAW99} applied this paradigm to the example of input-queued switches. Extensions of these policies can be found in Stolyar \cite{stolyar2004maxweight}. The MaxWeight and BackPressure policies have proved popular and, accordingly, have been generalized \cite{AKRSVW04,St04,ESP05,ShWi06,ShWi11,JMMT11,ShWi12}.
The defining feature of each generalization of MaxWeight and BackPressure is that it minimizes the drift of a Lyapunov function, which is then
used to prove the policies stability.

A further class of maximum stable Internet models are bandwidth sharing networks, first introduced by Roberts and Massouli\'e \cite{MaRo99}. Similar to MaxWeight and BackPressure, these policies are defined by an optimization and are maximum stable; however, unlike BackPressure, they are not constructed from a Lyapunov function. Stability proofs for these systems can be found in Bonald and Massouli\'e \cite{BoMa01} and Massouli\'e \cite{Ma07}.  Further progress on the large deviations and heavy traffic behaviour of proportional fairness can be found in Jonckheere, Lopez \cite{JoLo13} and Vlasiou, Zhang, Zwart \cite{VZZ14}. To the knowledge of this author, discussions applying these models to switch networks are first made by Shah and Wischik \cite{ShWi11} and Zhong \cite{Zh12} for the weighted $\alpha$-fair bandwidth allocation policies of Mo and Walrand \cite{MoWa00}.

This paper merges the MaxWeight and bandwidth allocation policies, proves their stability for single-hop switch networks and, for a proportionally fair case, we extend these results to multihop networks with fixed routing. In essence the paper leverages and generalizes results to switch networks that originate from bandwidth networks \cite{BoMa01,Ma07}.
Because of their substantially reduced queueing complexity, there are a number of significant structural advantages which have not previously been observed in switch systems.

\subsection{Policies and Results}
We now briefly describe the switch policies considered in this paper. A more formal description is given in Section \ref{sec3}.
In this paper, we define a switch policy as follows: given parameter $\alpha>0$ and concave functions $g=(g_j$: $j\in\mJ)$, the ($\alpha,g$)-switch policy chooses a random schedule which solves
\begin{align*}
&\text{maximize}\qquad\sum_{j\in\mJ} g_j(s_j)Q_j^\alpha\quad\text{over}\quad s\in\coS,
\end{align*}
where $Q=(Q_j:j\in\mJ)$ is the vector of queue sizes and where, allowing for randomization, the maximization is taken over a set of admissible schedules $\coS$. The interior of the set $\coS$ also gives the network\rq{}s \emph{stability region} -- the largest set of arrival rates for which a policy can stabilize the network. A policy which is stable for all arrival rates in $\coS$ is maximum stable. In order to place these policies on a par with the traditional MaxWeight policies, we show that these policies have the same maximum stability property as the MaxWeight. In Theorem \ref{MainThrm}, we provide a proof which is analogous to the proof of Bonald and Massouli\'e for weight $\alpha$-fair bandwidth policies \cite{BoMa01}. This is the first contribution of the paper.

Next we consider multihop networks with fixed routing. Here the BackPressure is the canonical maximal stable policy. The BackPressure policy essentially consists of two stages: the first stage, weights are determined for each link $j$ by comparing current and downstream stream queues sizes for each commodity/route and, a second stage, were we optimize these weights over the set of schedules.
Notice, for BackPressure to make a scheduling decision it must know the routes of all packets present at each queue. This is not practical when the routes processed a node are large in number and potentially unknown.
In contrast, we consider a proportionally fair optimization. This consists of solving an optimization of the form
\begin{align*}
&\text{maximize}\qquad\sum_{j\in\mJ} Q_j \log s_j\quad\text{over}\quad s\in\coS.
\end{align*}
Here $j=(n,n\rq{})$ represents a directed link between two nodes of a communication network and $Q_j$ is the number of packets present at $n$ which currently waiting to pass through link $j$. A schedule is then chosen whose mean is the solution to this optimization and a job is chosen at random from each scheduled queue.

We consider the fluid model for this fixed route multihop network.
We show that the proportional fair fluid model is maximum stable. We give two proofs of this result. The first exploits the structure of our random service discipline and use the properties of the logarithm to reduce our multi-class queueing system with fixed routing to a single-class system with Markovian routing. From this reduction, we can apply the fluid analysis of Massouli\'e \cite{Ma07} to prove fluid stability. We provide a second and direct proof that emphasized the policies implicit ability to balance the entropy between separate networked components. The fluid stability argument given is compared with the proof of Bramson \cite{Br96} for head-of-the-line processor sharing networks. The importance of these stability arguments, which originate from bandwidth networks and classical queueing networks, is in their consequences for switched systems.

\subsection{Contributions}

Let us discuss in more detail a few of the advantages of this Proportional Scheduler for multihop switched networks. Firstly, as we consider increasingly large communication networks, we expect the degree of nodes (routers) to stay bounded but the number of routes or destinations processed by a router to grow substantially. Thus the Proportional Scheduler, which maintains a queue per out-going link, requires a number of queues that is far far smaller than BackPressure, which maintains a queue for each route-destination.
Second, to make a scheduling decision with the BackPressure policy, we need to know the entire route or at least the destination in of each packet, while for the Proportional Scheduler knowing the next hop of the packet is sufficient to implement the policy.
Third, if we add new nodes to our network then, to implement BackPressure, the newly added node must be aware of the entire topology and route/class structure of the network, whilst the Proportional Scheduler only requires information about neighboring links. So, routing decisions can be determined by packets.
 Fourth, when different schedules within the network can be implemented independently -- here we imagine separate non-interfering network components --, for BackPressure, messages must be shared between the separate components in order to determine the differences in queue sizes. However, for the Proportional Scheduler, the optimization relevant to the policy will decompose into an independent optimization for each component which can then be solved separately.
Finally, as a general rule, there are almost no single-hop communication systems. Communication systems consist of a large number of interconnected components. Single-hop switch networks are often considered in performance analysis for reasons of tractability. This statement is particularly true for wireless systems. However, stability in a single-hop system does not in general imply stability in a corresponding multihop system. See Lu and Kumar \cite{LuKu91}, Rybko and Stolar \cite{RySt92} and Bramson \cite{Br08} for instability examples which can easily be adapted to switch systems. Notable, recent work of Dieker and Shin \cite{dieker2013local} considers schemes which extend local/single-hop stability results to global stability results. These allow for multi-class routing. However, like BackPressure, these policies require route information from each packet to make scheduling decisions. To the best of knowledge of this author, this is the first general proof of maximum stability for a multihop switch network where routing information is not required in order to make a scheduling decision.

We now briefly discuss extensions and future consequences of this analysis. The Proportional Scheduler which we consider applies a processor-sharing discipline within each queue, meaning that the scheduled packet from a queue is chosen at random. However, it would be desirable to consider first-in-first-out queueing. In this way, packets will be served by the network in sequence and jitter effects will be reduced. Further, in practice, communication implement FIFO queueing. Surprisingly, the Proportional Scheduler is provably a maximum stable for FIFO queueing disciplines.  The proof of this result is substantially more involved than the short proof that we can provide for the processor-sharing case. The FIFO case follows due to underlying connections between proportional fairness and quasi-reversible queueing systems -- see \cite{Ma07,ShWaZh12} for some related discussion. This analysis leads to further structural advantages namely product-form resource-pooling effects and better delay scaling complexity compared to BackPressure. We refer the interested reader to the preprint \cite{BDW14}.

In a summary, the contributions of this work are as follows:

\begin{itemize}
\item We combine the set of MaxWeight policies and bandwidth allocation policies, to provide a new class of switch policies and we prove maximum stability for these polices.
\item Using a random discipline within queues, we consider a Proportional Scheduler for multihop networks with fixed routing. We give two proofs of fluid stability for this system.
\item We emphasis several important observations of this last result for switch networks:
\renewcommand{\labelitemii}{$\circ$}
\begin{itemize}
\item the policy only needs to maintain a queue for each outgoing link;
\item packets can be routed by only knowing each packets next hop;
\item network nodes do not need to know their network topology;
\item message do not need to be sent to calculate weights and so scheduling decisions can be completely decomposed between independently functioning components.
\end{itemize}
\item We provide  the first general proof of maximum stability for a multihop switch network where routing information is not required in order to make a scheduling decision.
\end{itemize}
Informally speaking, because the Proportional Scheduler does not need to know the route structure of the network to make a scheduling decision, while BackPressure does, the policy\rq{}s implementation scales better with network size.

\subsection{Organization}
The remaining sections of the paper are organized as follows. In Section \ref{sec2}, we describe the network and queueing process for both single-hop and multihop networks. In Section \ref{sec3}, we define the $(\alpha,g)$-switch policies, proportional fairness and we relate these to existing policies, namely, MaxWeight, BackPressure, utility optimizing bandwidth allocation policies, and weight $\alpha$-fairness. We also define the fluid model associated with each of these systems. In Section \ref{sec4}, we present and discuss the main results of this paper, Theorem \ref{MainThrm} and Theorem \ref{MainThrm2}. In Section \ref{sec5}, we provide a proof of Theorem \ref{MainThrm}. This consists of characterizing the fluid limit of the $(\alpha,g)$-policy, proving stability of that fluid system and then using this to prove positive recurrence of the prelimit process. In Section \ref{sec6}, we give two proofs of Theorem \ref{MainThrm2}. The first involves a reduction to a proof of Massouli\'e \cite{Ma07}. The second proof can be argued from the fluid analysis of Bramson \cite{Br96}.

\section{Switch Network Notation}\label{sec2} 
We define single-hop and multihop switch networks. These discrete-time queueing networks have restrictions on which queues can be served simultaneously.

The following notation is used both for single-hop and multihop networks.
We assume that time is slotted, that is each time, $t$, belongs to the positive integers, $\bZ_+$.
We let the finite set $\mJ$ index the set of queues or links of a network.
We let the finite set $\mS$ be the set of schedules.
Each schedule $\sigma=(\sigma_j:j\in\mJ)\in\mS$ is a vector in $\bZ_+^\mJ$ where $\sigma_j$ gives the number of jobs that will be served from queue $j$ under schedule $\sigma$. We assume the zero vector, $\textbf{0}$, belongs to $\mS$. We let $\sigma_{\max}$ give the maximum value of $\sigma_j$ for $\sigma\in\mS$.
We let $\coS$ be the convex combination of schedules in $\mS$.
We assume $\coS$ has a non-empty interior.
For real numbers $q,s\in\bR$, we define $q\wedge s=\min\{q,s\}$. For vectors $q,s\in\bR^\mJ$, we define $q\wedge s= (q_j\wedge s_j: j\in\mJ)$.
For a vector of queue sizes $Q\in\bZ_+^\mJ$, we define
\begin{align*}
\mS_Q = \{ \sigma \wedge Q : \sigma \in \mS \}.
\end{align*}
Given that not all jobs can be served from a queue $j$ when $Q_j < \sigma_j$, the set $\mS_Q$ gives the set of schedules available given the vector of queue sizes $Q$. We let $<\!\mS_Q\! >$ be the convex combination of points in $\mS_Q$.

\subsection{Single-hop Switched Network}

For a \emph{single-hop network}, once a packet has been served at its queue it departs the network.

We let $a(t)=(a_j(t):j\in\mJ)\in\bZ_+^\mJ$ be the number of arrivals occurring at each queue
at time $t\in\bZ_+$. We assume $\{ a(t) \}_{t=0}^\infty$ is a sequence of independent identically distributed random vectors
with finite mean $\bar{a}\in(0,\infty)^\mJ$ and finite second moment $\bE\big[ a_j(t)^2\big] \leq K$ for $j\in\mJ$.\footnote{The assumption of finite variance is not essential; however, it allows for bounds more convenient for our proofs.}

We let $Q(0)=(Q_j(0):j\in\mJ)$ be the number of jobs in each queue at time $t=0$.
From a sequence of schedules $\{ \sigma(t) \}_{t=1}^\infty$, we can define the queue size vector $Q(t)=(Q_j(t):j\in\mJ)$ by
\begin{equation}
Q_j(t+1)= Q_j(t) - \sigma_j(t+1)  + a_j(t+1),
\end{equation}
for $j\in\mJ$, and $t\in \bN$. For positive queue sizes, we require
that $\sigma_j(t+1)\in \mS_{Q(t)},$
for all $t\in\bZ_+$.\footnote{Observe, this choice of notation is equivalent to defining the queueing process by $Q_j(t+1)= [ Q_j(t) - \sigma_j(t+1)]_+  + a_j(t+1)$.}  It is well known that queue sizes cannot be stabilized when the arrival rates $\bar{a}$ lie outside the set $< \! \mS \! >$, for instance, see  \cite{TaEp92}. For this reason, we give the following definition.

Given $\{Q(t)\}_{t=0}^\infty$ defines a Markov chain, we say the queue size process is \emph{maximum stable} if it is positive recurrent whenever the vector of arrival rates, $\bar{a}$, belongs to the interior of $\coS$.

\subsection{Multihop Switched Network}
In a multihop network we allow packets to visit multiple queues within the network before departing. We will develop our  notation in a similar way to Tassiulas and Ephremides \cite{TaEp92}.

As in the previous section, we let the finite set $\mJ$ index the set of queues or links of a network.
We let $\mN$ define the set of nodes of a network.
We consider each link $j\in\mJ$ to be a directed edge in this network $j=(n,n\rq{}) \in \mN\times \mN$.
We let $\mR$ be a set of routes through the network.
Here each $r=(n_1^r,...,n_{k_r}^r)\in\mR$ is an ordered set of nodes such that $j_k^r:=(n^r_k,n^r_{k+1})\in \mJ$ for $k=1,...,k_r-1$. A route $r$ packet served at node $n_k$ must next visit node $n_{k+1}$ using link $j^r_k$. We apply the notation $j\in r$ if link $j$ is part of route $r$. Further, for each $j\in r$, we let $j^r_-$ and $j^r_+$ denote the previous (upstream) and the next (downstream) link on route $r$.

We let $a(t)=(a_{jr}(t):j\in\mJ, r\in\mR)\in\bZ_+^{\mJ\times\mR}$ be the number of external arrivals occurring at the ingress $j=j^r_1$ of each route $r$
at each time $t\in\bZ_+$. Note $a_{jr}(t)=0$ unless the queue is the first on its route, $j=j^r_1$. We assume $\{ a(t) \}_{t=0}^\infty$ is a sequence of independent identically distributed random vectors
with finite mean $\bar{a}=(\bar{a}_{jr}:r\in\mR, j\in\mJ)\in\bR_+^{\mJ\times\mR}$ and finite second moment $\bE\big[ a_{jr}(t)^2\big] \leq K$ for all $r\in\mR$ $j\in\mJ$. We define the load on queue $j\in\mJ$ by
\begin{equation}
\bar{a}_j = \sum_{r: j\in r} \bar{a}_{jr}.
\end{equation}
Also we let $\bar{a}_r:=\bar{a}_{j_1^r r}$  be the arrival rate of route $r$.

We let $Q_j(t)$ denote the number of packets who must next be served along link $j$. We let $X_{jr}(t)$ be the number of route $r$ packets who will next be served over link $j$ at time $t$. Thus, it holds that
\begin{equation}
Q_j(t)=\sum_{r: j\in r} X_{jr}(t).
\end{equation}

We must describe the queueing dynamics for the queue size process $X(t)=(X_{jr}(t):j\in\mJ, r\in\mR)$. In particular --  in constrast to single-hop networks which are single-class -- because we are considering a multi-class queueing network, once the links have been scheduled we must also decide which routes
will be served. For the BackPressure policy, this would involve comparing the queue size for each route at the queue and at the next downstream node. Shortly, we will explain how this decision is made as part of a Proportional Scheduler. 
In any case, given a schedule $\sigma(t)\in\mS_{Q(t)}$ is chosen at time $t$, we let $\xi(t)=(\xi_{jr}(t):  j\in\mJ, r\in\mR)$ denote the number of route $r$ packets served from queue $j$ under schedule $\sigma(t)$. For $\xi(t)$, the following constraints must hold
\begin{align*}
\sum_{r: j\in r}  \xi_{jr} &(t) = \sigma_j(t),  \qquad
\xi_{jr}(t) \leq X_{jr}(t).
\end{align*}
That is the number of jobs served from each class is the number scheduled and we cannot schedule more jobs from each class than there are at the queue.
Given this, the queueing dynamics of a multihop switch network are as follows
\begin{equation}
X_{jr}(t+1) = X_{jr}(t) + a_{jr}(t+1) - \xi_{jr}(t+1) + \xi_{j\rq{}r}(t+1)
\end{equation}
where for each $r\in\mR$ and each link $j\in r$ its upstream link $j\rq{}$.\footnote{If no upsteam link exists for $j$ then we set $\xi_{j\rq{}r}(t+1)=0$.} Notice that once a packet has been served at its queue it joins the next queue on its route. We call this a \emph{multihop network with fixed routing}. It is well known that queue sizes cannot be stabilized when the arrival rates $\bar{a}$ lie outside the set $< \! \mS \! >$, for instance, see  \cite{TaEp92}. For this reason, we give the following definition.

Given the process $\{X(t)\}_{t=0}^\infty$ is a Markov chain, we say the queue size process is \emph{maximum stable} if it is positive recurrent whenever the vector of arrival rates, $(\bar{a}_j: j\in\mJ)$, belongs to the interior of $\coS$.

\section{Switch Policies}\label{sec3}
We now define the main switching policies considered in this paper. For single-hop networks, we introduce the ($\alpha,g$)-switch policies and, for multihop networks, we consider the Proportional Scheduler.

\subsection{The ($\boldsymbol\alpha$, g)-Switch Policy}
We let $\alpha$ be a positive real number and, for each $j\in\mJ$, we let $g_j:\bR_+\rightarrow \bR$  be a strictly increasing, differentiable, strictly concave function. Given a queue size vector $Q(t)=(Q_j(t):j\in\mJ)$ at time $t$, we define $\bar{\sigma}(t+1)=(\bar{\sigma}_j(t+1):j\in\mJ)$ to be a solution to the optimization
\begin{subequations}\label{MW-af}
\begin{align}
&\text{maximize} &&\sum_{j\in\mJ} g_j(s_j)Q_j(t)^\alpha\label{MW-af-1}\\
&\text{over} && s\in\coqS{t}.\label{MW-af-3}
\end{align}
\end{subequations}


In general, $\bar{\sigma}(t+1)$ need not belong to the set of schedules $\mS_{Q(t)}$. However, $\bar{\sigma}(t)$ is a convex combination of points in $\coqS{t}$. Thus we let $\sigma(t+1)$ be a random variable with support on $\mS_{Q(t)}$ and mean $\bar{\sigma}(t+1)$. The ($\alpha,g$)-switch policy scheduling policy is the policy that chooses schedule $\sigma(t)$ at each time $t\in\bN$.

To be concrete, we assume the random variables $\sigma(t)$ are, respectively, a function of $Q(t-1)$ and an independent (uniform) random variable. 
This ensures that
the queue size process $\{ Q(t)\}_{t=0}^\infty$ associated with the ($\alpha,g$)-switch policy scheduling policy is a discrete-time Markov chain. Further, the constraints \eqref{MW-af-3} ensure that a schedule never exceeds the queue it serves.

\subsubsection*{MaxWeight, Utility optimization and $\alpha$-fairness}
We now briefly compare the class of policies above with the MaxWeight and network utility optimizing allocations.
A MaxWeight-$\alpha$ allocation policy would be the special case of the optimization problem \eqref{MW-af} when $g_j(s_j)=s_j$,
\begin{align}
\text{maximize} \quad \sum_{j\in\mJ} s_jQ_j(t)^\alpha\quad \text{over}\quad  s\in\mS_{Q(t)}.
\end{align}
A network utility maximizing allocation in the sense of \cite{Ke97} would correspond to an optimization of the form
\begin{equation}
\text{maximize} \quad \sum_{j\in\mJ} g_j\left( \tfrac{s_j}{Q_j(t)} \right) Q_j(t)\quad \text{over}\quad  s\in\coS.
\end{equation}
Notice in this policy, terms ${s_j}/{Q_j(t)}$ are applied to concave $g_j$, whilst the $(\alpha,g)$-policy just applies $s_j$. The ${s_j}/{Q_j(t)}$  term is included for bandwidth network because it applies processor-sharing system to each transfer. However, since a switch-network applies its discipline to each component of the switch, we do not divide by $Q_j$. A further technical point is that, when compared with the $(\alpha,g)$-policies, it is significantly harder to prove the stability of the above bandwidth network.

A further interesting case is the $\alpha$-fair family of bandwidth allocation policies, as first introduced by Mo and Walrand \cite{MoWa00}. For $\alpha>0$ and $\alpha\neq 1$, these are defined as follows
\begin{align*}
&\text{maximize} \quad\sum_{j\in\mJ}\frac{ s_j^{1-\alpha}}{1-\alpha}Q_j(t)^{\alpha}\quad \text{over}\quad  s\in\coS.
\end{align*}
We remark that the ($\alpha,g$)-switch policy coincides with the MaxWeight-$\alpha$ policy when $g_j(s_j)=s^{1-\alpha}_j/(1-\alpha)$. For switch-networks, the $\alpha$-fair policy is first considered by Shah and Wischik \cite{ShWi11}.

\subsubsection*{Fluid Model}
The fluid model associated with the $(\alpha,g)$-switch policy is a positive, absolutely continuous process $q(t)= (q_j(t): j\in\mJ)$, $t\in\bR_+$, where, for $j\in\mJ$ and for almost every\footnote{By \emph{almost every}, we mean on all points except for a set of Lebesgue measure zero.} time $t\in\bR_+$,  if $q_j(t)>0$ then
\begin{align}\label{FluidEqns}
\frac{d q_j}{dt} & = \bar{a}_j - {\sigma}^*_j(q(t)).
\end{align}
Here ${\sigma}^*(q)$ solves the optimization

\begin{align}
&\text{maximize}\qquad\sum_{j\in\mJ} g_j(s_j)q_j^\alpha\qquad \text{over}\qquad s\in\coS.\label{MW-sq}
\end{align}
Note that we optimize over $\coS$ rather than $<\! \mS_q \!>$.

We say that a fluid model is \emph{stable} if there exists a time $T>0$ such that,
for every $\{q(t)\}_{t\in\bR_+}$ satisfying \eqref{FluidEqns} and with $||q(0)||_1=1$,
\begin{align}
&\qquad\qquad\qquad q_j(t)=0, && j\in\mJ,
\end{align}
{ for all } $t\geq T$. Here and hereafter, $||\cdot||_1$ is the $L_1$-norm.


\subsection{Proportional Scheduler}\label{MH-PF-sec}
We now describe a scheduling policy for multihop networks.
Given a vector of link queue sizes $(Q_j(t): j\in\mJ)$, the Proportional Scheduler for multihop networks is defined as follows

\begin{description}\item[PS1.]  Over set of schedules $ <\! \mS_{Q(t)}\! >$ solve the optimization
\begin{subequations}\label{MH-PF}
\begin{align}
&\text{maximize}&& \sum_{j\in\mJ} Q_j(t) \log (s_j) \\
&\text{over}&& s\in <\! \mS_{Q(t)}\! >.
\end{align}
\end{subequations}
As previously, let $\sigma(t+1)$ be a random variable on $\mS_{Q(t)}$ whose mean solves this optimization.
\item[PS2.] From each queue $j\in\mJ$, serve $\sigma_j(t+1)$ packets selected uniformly at random from $Q_j(t)$ packets at the queue. These then join their next downstream node as determine by their route class.
\end{description}
Notice we do not need to know the routes used by packets in order to define the above policy. We assume that, after service at a queue, packets go to the next hop on their respective routes. The optimization \eqref{MH-PF} is due to Kelly \cite{Ke97} and is called the proportional fair optimization. We refer to the policy as the \emph{Proportional Scheduler}.\footnote{Since fairness is not, as such, an objective in our scheduling discipline, we omit the use of the word fair.}

\subsubsection*{The BackPressure Policy}
We now briefly compare the Proportional Scheduler with the BackPressure policy.
For multihop networks, the principle maximal stable policy has been the BackPressure policy.

Given the vector of queue sizes $(X_{jr}(t): j\in\mJ, r\in\mR)$, this policy is defined by the following three steps
\begin{description}\item[BP1.] For each link $j=(n,n')\in\mJ$ directed out of each node $n\in\mN$, calculate weights by comparing with upstream queue lengths,
\begin{equation}\label{BPwOpt}
w_j(X(t))= \max_{r: j\in r} \left\{   X_{jr}(t) - X_{j_+^rr}(t) , 0 \right\},\footnote{Here if there is no next link after $j$ on route $r$ then we set $X_{j_+^rr}=0$.}
\end{equation}
and let $r^*_j(X(t))$ be the solution to this maximization.
\item[BP2.] Over set of schedules $\mS$, solve the optimization
\begin{equation}\label{BPOpt}
\max_{\sigma\in\mS} \;\; \sum_{j\in\mJ} \sigma_j w_j(X(t)),
\end{equation}
and let $\sigma^*(X(t))$ be the solution to this optimization.
\item[BP3.] If $w_j(X(t))>0$, at the next time instance schedule $\sigma_j^*(X(t))$ packets from route $r_j^*(X(t))$ from each link $j\in\mJ$, else, do not schedule any packets on link $j$.
\end{description}
MaxWeight is the special case of BackPressure when we consider a single-hop network. Notice in the first step above, information must be exchanged along links to make a queue size comparison \eqref{BPwOpt}.

In contrast to BackPressure, the Proportional Scheduler does not maintain a queue for each route. It maintains a queue for each link. Further, the Proportional Scheduler does not compare queue sizes with downstream queues. Thus the policy is, in a certain sense, more decentralized. In the Proportional Scheduler packets are selected at random from the queue. In most communications systems packets are served according to a first-in-first-out discipline. We remark that it is possible to prove maximum stability results when the randomized queueing discipline PF2 is replaced by a FIFO queueing discipline. However, the proof is significantly more technical in nature. We refer the read to the forthcoming article for a proof \cite{BDW14}.

\subsubsection*{Tree Network Example}
To emphasis the advantages of the queueing structure of our policy. We consider a tree network where each node of degree $d$ and diameter $D$. See Figure \ref{Tree} for a tree of degree $d=3$ and diameter $D=6$. BackPressure, as described above, requires each node to maintain a queue for each route passing through it. Thus, one can verify, that for this network the central node of a tree must maintain a number of queues given by
\begin{equation*}\label{ExBPbound1}
d(d-1)^{(D -1)}.
\end{equation*}
\begin{figure}[h!]
  \centering
\includegraphics[width=0.45\textwidth]{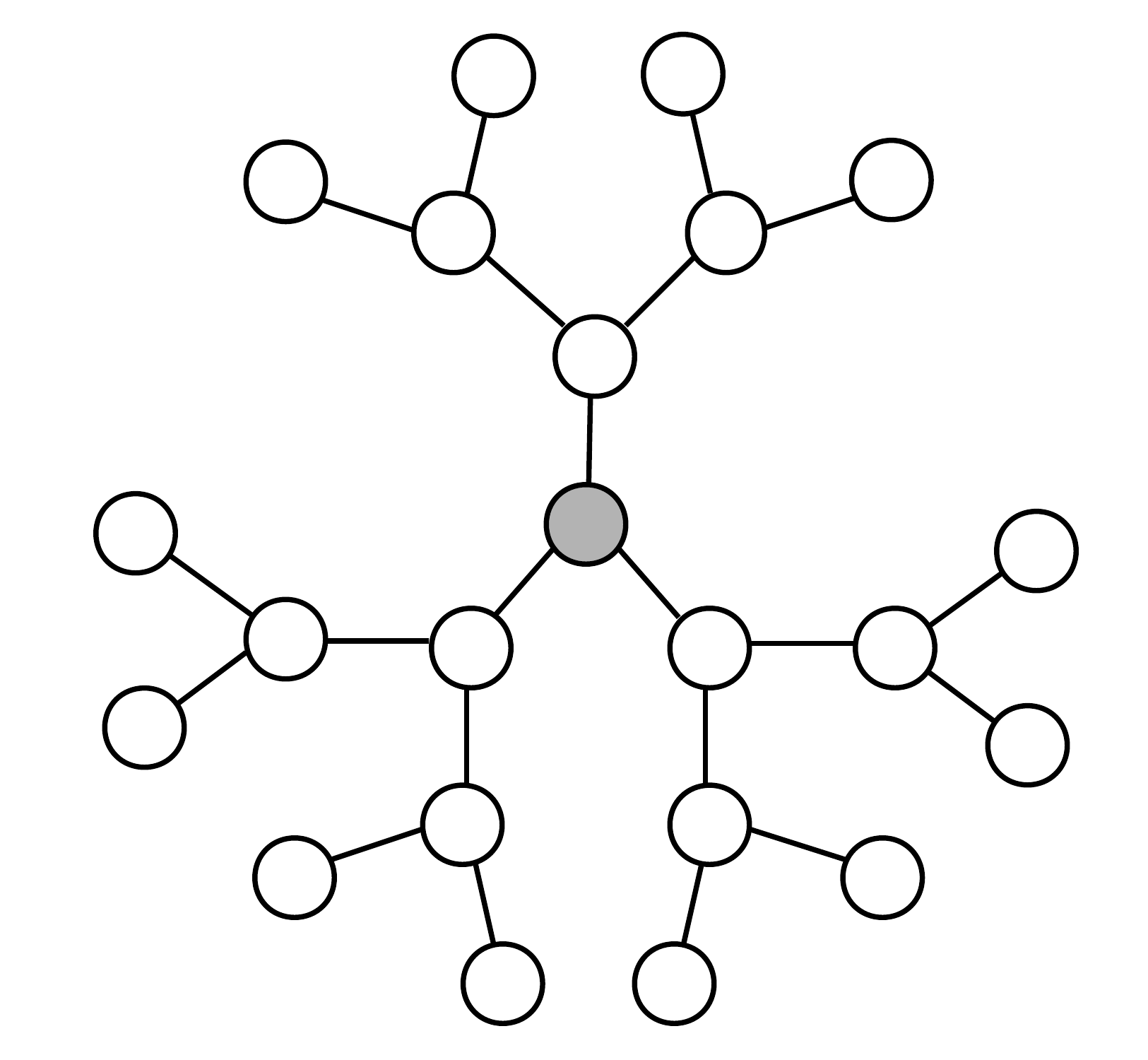}
\caption{A tree network of depth $d=3$ and diameter $D=6$. There are $12$ leaf nodes and thus $132= 12\cdot 11$ directed routes between leaf nodes. \label{Tree} }
\end{figure}
For this example all routes with the same destination can be merged once they intersect. For this reason, BackPressure can be implemented slightly more efficiently in each node by storing a queue for each destination, however even in this case the number of queue\rq{}s required grows as
\begin{equation*}\label{ExBPbound2}
d(d-1)^{({D}/{2}-1)}.
\end{equation*}
In otherwords, to implement BackPressure the nodes of this network must maintain a number of queues that grows exponentially with the diameter of the network. In general, one expects the number of queues required to implement BackPressure to grow substantially with the network\rq{}s size.
However, we observe the Proportional Scheduler maintains a queue for each outgoing link from a node. That is for each node $n$ the policy maintains a queue for each $j=(n,n\rq{})\in\mJ$ rather than each link-route pair $(j,r)\in\mJ\times\mR$. Thus, in our example, a Proportional Scheduler would require a number of queues given by $d$. This is significantly smaller than the number required by BackPressure. In general, for the Proportional Scheduler to be implemented each node requires a queue for each neighboring node, which is smaller than the network size and is often bounded as we increases the size of a network.

\subsubsection*{Fluid Model}

The fluid model associated with the multihop network operating under the Proportional Scheduler, is a positive, absolutely continuous process $x(t)=(x_{jr}(t) : j\in\mJ, r\in\mR)$ and $q(t)= (q_j(t): j\in\mJ)$, $t\in\bR_+$, where for almost every\footnote{By \emph{almost every}, we mean on all points except for a set of Lebesgue measure zero.} time $t\in\bR_+$,
\begin{align}\label{hopfluid}
&\frac{d x_{jr}}{dt} =
\frac{x_{j^r_- r}(t)}{q_{j^r_-}(t)} \sigma^{*}_{j^r_-} (q(t)) - \frac{x_{j r}(t)}{q_{j}(t)} \sigma^{*}_{j}(q(t)),
\end{align}
where ${\sigma}^*(q)$ solves the optimization
\begin{align}
&\text{maximize}\qquad\sum_{j\in\mJ} q_j \log s_j\qquad \text{over}\qquad s\in\coS.\label{MW-PF}
\end{align}
Above, we apply the convention that
\begin{equation}\label{sumxq}
q_j(t)=\sum_{r: j\in r} x_{jr}(t).
\end{equation}
Further, in the equation \eqref{hopfluid}, if the queue $j$ is the first queue on route $r$, i.e. $j=j^r_1$  and $j^r_-=j^r_0$, then we apply the convention that  \footnote{This is similar to the proof of Bramson \cite{Br96} for Head of the line processor sharing networks.
Bramson introduces an additional queue for each route as a proof device.
Whilst here, the introduction of $j^r_0$ is purely for notational and conceptual convenience.}
\begin{equation}
\frac{x_{j^r_0 r}(t)}{q_{j^r_0}(t)} \sigma^{*}_{j^r_0} (q(t))=\bar{a}_{jr}.\label{Convension}
\end{equation}
In this way, we account for external arrivals.

\section{Main results} \label{sec4}
The main results of this article is to prove maximum stability for the $(\alpha,g)$-switch policies in single-hop networks and to prove maximum stability for the Proportional Scheduler  in multihop networks with fixed routing. These results are stated as follows

\begin{theorem}\label{MainThrm}
For the ($\alpha,g$)-switch policy,
if the vector of average arrival rates $\bar{a}=
(\bar{a}_j: j\in\mJ)$ belongs to the interior of $\coS$
then the queue size process $\{Q(t)\}_{t=0}^\infty$  is positive recurrent.
\end{theorem}

\begin{theorem}\label{MainThrm2}
For the Proportional Scheduler applied to a multihop switched network,
if the vector of average arrival rates $\bar{a}=(\bar{a}_j : j\in\mJ)$ belongs to the interior of $\coS$
then the fluid model of this system is stable.
\end{theorem}

 The first result is proven by appropriately modifying the Lyapunov argument of Bonald and Massouli\'e for weighted $\alpha$-fair bandwidth networks. For the second result, we give two proofs. The first proof shows how the result can be demonstrated within the framework of Massouli\'e \cite{Ma07}. Essentially Massouli\'e's result proves maximum stability of proportional fairness in single-class networks with probabilistic routing (i.e. Jackson type \cite{Ja63}), whilst we wish to prove maximum stability of proportional fairness for multiclass networks with fixed routing (i.e. Kelly type \cite{BCMP75,Ke79}).
 We note that, due to the structure of the logarithm used to define the Proportional Scheduler, the latter problem can be reduced to the former. The second proof observes the result can be proven within the framework of Bramson \cite{Br96}. Bramson shows the stability of fluid models of networks of processor sharing queues. We show by the differentiability properties of the proportionally fair objective that the fluid model for the Proportional Scheduler has essentially the same behavior as a network of processor sharing queues. Both these proofs and a number of consequences and conjectures rely on deeper connections with reversible queueing systems. These are described in the forthcoming paper \cite{BDW14}.

Theorem \ref{MainThrm} shows that the class of $(\alpha,g)$-switch policies have the same stability properties as the MaxWeight policies. This, like MaxWeight, is proven for single-hop networks; however, for multiclass queueing networks, results that hold for single-hop systems do not naturally extend when directly applied multihop systems. See Chapter 3 of Bramson \cite{Br08} for various counter examples of this type. This point is particularly poignant for switch networks where there has been a great deal of recent progress on the stability of single-hop switch networks, \cite{JiWa10,ShSh12}. Theorem \ref{MainThrm2} gives the first result where stability of a single-hop switch network implies stability of its multihop counterpart without requiring additional routing information. See Dieker and Shin \cite{dieker2013local} for some recent related work.
The result relies on the specific properties of the proportional fair optimization. Nonetheless the hope is this initial result will lead to general methods for extending stability results from single-hop to multihop networks.

There are a number of immediate advantages of the proportional fair optimization in comparison to BackPressure. We only need to maintain a queue for each outgoing link, whilst BackPressure requires a vastly more complex data structure. BackPressure must know the route class of each packet within the network to make a routing decision whilst for our policy packets can be routed by only knowing each packets next hop.
Messages do not need to be sent to calculate weights and so scheduling decisions can be completely decomposed between independently functioning components. That is, if $\mS$ and $\mJ$ can be split into independent components $\prod_{k} \mS_k$ and $\cup_k \mJ_k$, then
\begin{equation*}
\max_{s\in <\! \mS\! >} \sum_{j\in\mJ} Q_j(t) \log s_j = \sum_k \max_{s\in <\! \mS_k\! >} \sum_{j\in\mJ_k} Q_j(t) \log s_j.
\end{equation*}
Here we imagine each $\mS_k$ corresponds to the constraints of an independent network component, for example, an input queued switch for each $k$. The above statement shows that, without messaging, we can solve this optimization separately for each component, implement the decision and still guarantee stability.

Informally speaking, for the Proportional Scheduler, does not need to know the route structure of the network to make a scheduling decision, whilst BackPressure does. For this reason, the proportional fair implementation, scales far better with network size.

\section{Proof of Theorem 1}\label{sec5}
The proof of Theorem  \ref{MainThrm} follows the fluid limit approach of Dai \cite{Da95}. We prove that a certain fluid model is satisfied by the limit of a sequence Markov chains using $(\alpha,g)$-switch policy. We prove that the associated fluid model with the $(\alpha,g)$-policy is stable.
We then give a proof of positive recurrence of the original Markov chain.

\subsection{Fluid model}
In this section, we state fluid model equations and the fluid limit associated with the ($\alpha,g$)-switch policy.

Let $\{ Q^{(c)} \}_{c\in\bN}$ be a sequence of versions of our queueing process for the ($\alpha,g$)-switch policy, where $||Q(0)||_1=c$. We define
\begin{align}
\bar{Q}^{(c)}(t)& =\frac{Q^{(c)}(\lfloor ct\rfloor)}{c},
\end{align}
for $t\in\bR_+$.  The following result formalizes a fluid model $q$, \eqref{FluidEqns}, as the limit of $\{ \bar{Q}^{(c)} \}_{c\in\bN}$. In informal terms, it states that the only possible limit of the sequence $\{ \bar{Q}^{(c)} \}_{c\in\bN}$ as $c\rightarrow\infty$ is a process $q$ satisfying \eqref{FluidEqns}.

\begin{proposition}[Fluid Limit]\label{FluidLimThrm}
The sequence of stochastic processes $\{ \bar{Q}^{(c)}\}_{c\in\bN}$ is tight\footnote{Recall a sequence of random processes $\{ \bar{Q}^{(c)}\}_{c}$ is tight if every subsequence of $\{ \bar{Q}^{(c)}\}_{c}$ has a weakly convergent subsequence.} with respect to the topology of uniform convergence on compact time intervals. Moreover,  any weakly convergent subsequence of  $\{ \bar{Q}^{(c)}\}_{c\in\bN}$ converges to a Lipschitz continuous process satisfying fluid equation \eqref{FluidEqns}.
\end{proposition}

\noindent The statement of Proposition \ref{FluidLimThrm} is somewhat technical. However, the main point is that we can compare the queueing process to a tractable fluid model, $q$ satisfying equations \eqref{FluidEqns}. The proof of Proposition \ref{FluidLimThrm} is somewhat standard. We refer the reader to Robert \cite{Ro10}. We now analyze the stability of the fluid models.

\subsection{Fluid Stability}\label{sec7}
We consider a process $\{q(t)\}_{t\in\bR_+}$, that satisfies the fluid limit equations \eqref{FluidEqns}. In the following theorem, we show that these fluid solutions are stable in the sense that they hit the zero state in finite time. This result will be sufficient to prove positive recurrence of the ($\alpha,g$)-switch policy.

 \begin{proposition}[Fluid Stability]\label{FluidStabProof}
Given $(\bar{a}_j : j\in\mJ)$ belongs to the interior of $\coS$. There exists a time $T>0$ such that,
for every fluid model $\{q(t)\}_{t\in\bR_+}$ satisfying \eqref{FluidEqns} and with $||q(0)||_1=1$,
\begin{align}
&\qquad\qquad\qquad q_j(t)=0, && j\in\mJ,
\end{align}
{ for all } $t\geq T$.
\end{proposition}

The main idea is to consider the gradient of the tangent line of the ($\alpha,g$)-switch policy objective between two points: the arrival rate and the optimal solution, see \eqref{rhocond} below. Integrating this obtains a Lyapunov function \eqref{MW-Lya}. This idea is used by Bonald and Massouli\'e \cite{BoMa01} in their analysis of weighted $\alpha$-fair bandwidth sharing networks. For switch networks, Proposition \ref{FluidStabProof} follows analogously.

\begin{proof}
We define $G_q(s)$ to be the objective of the ($\alpha,g$)-switch policy optimization,
\begin{equation}
G_q(s)=\sum_{j\in\mJ} g_j(s_j) q_j^\alpha.
\end{equation}
Recall that in our fluid equations
\begin{equation*}
{\sigma}^*(q(t))\in\argmax_{s\in\coS}\; G_{q(t)}(s).
\end{equation*}
Any vector $\rho$ belonging to the interior of $\coS$ is not optimal. Thus  $G_{q(t)}(\rho)< G_{q(t)}({\sigma}^*(q(t)))$. As $G_{q(t)}(\cdot)$ is strictly concave,  $G_{q(t)}(\cdot)$ must be increasing along the line connecting $\rho$ to ${\sigma}^*(q(t))$. In other words, for all $\rho$ in the interior of $\coS$ and for  $q(t)\neq 0$,
\begin{equation}\label{rhocond}
\Big({\sigma}^*(q(t))-\rho\Big) \cdot \nabla G_{q(t)}(\rho) > 0.
\end{equation}
Here $\nabla G_{q}(\rho)= ( g'_j(\rho_j) q_j(t)^\alpha : j\in\mJ)$.
Since $\bar{a}$ belongs to the interior of $\coS$, there exists $\epsilon>0$ such that $ (1+\epsilon) \bar{a}\in \coS$. We define $\rho=(1+\epsilon)\bar{a}$. In this case, we can re-express the inequality \eqref{rhocond} as follows
\begin{equation}\label{ineq}
\sum_{j\in\mJ} \Big(\bar{a}_j-{\sigma}_j^*(q(t)) \Big) g_j'(\rho_j) q_j(t)^\alpha \leq -\epsilon \sum_{j\in\mJ} g_j'(\rho_j) q_j(t)^\alpha .
\end{equation}
As $\frac{d q_j}{dt}=\bar{a}_j-{\sigma}_j^*(q(t))$, we define the Lyapunov function
\begin{equation}\label{MW-Lya}
L(q)= \sum_{j\in\mJ} g'(\rho_j) \frac{q_j^{1+\alpha}}{1+\alpha},
\end{equation}
$q\in\bR_+^\mJ$. The function $L(q)$ is positive, and $L(q)=0$ iff $q_j=0$ for all $j\in\mJ$. We now observe
\begin{align}
\frac{dL(q(t))}{dt}&=\sum_{j\in\mJ} \Big(\bar{a}_j-{\sigma}_j^*(q(t)) \Big)g_j'(\rho_j) q_j(t)^\alpha \label{diffinequ0}\\
& \leq  -\epsilon \sum_{j\in\mJ} g_j'(\rho_j) q_j(t)^\alpha. \label{diffinequ}
\end{align}
The equality holds by the chain rule and the inequality holds by \eqref{ineq}.

We define two norms on $\bR_+^\mJ$
\begin{align*}
 ||q||_{1+\alpha} = ( L(q) )^{\frac{1}{1+\alpha}},\quad  || q||_{\alpha} = \Big( \sum_{j\in\mJ} g_j'(\rho_j) q_j(t)^\alpha\Big)^{\frac{1}{\alpha}}.
\end{align*}
By the Lipschitz equivalence of norms, there is a constant $\gamma>0$ such that $$ \gamma ||q||_{1+\alpha} \leq ||q||_\alpha,$$ for all $q\in\bR_+^\mJ$.\footnote{ Note,  $||q||_{1+\alpha} \leq (1+\alpha)^{-\frac{1}{1+\alpha}}|\mJ| \max_j g'_j(\rho_j) q_j $ and also note that  $\max_j g'_j(\rho_j) q_j\leq ||q||_{\alpha}$. So, for instance, we can take $\gamma=(1+\alpha)^{\frac{1}{1+\alpha}} | \mJ|^{-1}$.} Applying this observation to the inequality \eqref{diffinequ}, we see that
\begin{equation}\label{Lineq}
\frac{dL(q(t))}{dt} \leq -\epsilon \gamma^{\alpha}  L(q(t))^{\frac{\alpha}{1+\alpha}}.
\end{equation}
Observe, by the above inequality, if $L(q(T))=0$ for any differentiable point $T$ then  $L(q(t))=0$ for all $t\geq T$. Now, whilst $L(q(t))>0$, we have from \eqref{Lineq} that
\begin{align*}
&L(q(t))^\frac{1}{1+\alpha}-L(q(0))^\frac{1}{1+\alpha} \\
=& \int_0^t (1+\alpha)^{-1} L(q(t))^{\frac{-\alpha}{1+\alpha}}\frac{dL(q(t))}{dt} dt\\
\leq &  -\epsilon (1+\alpha)^{-1} \gamma^{\alpha} t.
\end{align*}
Rearranging this expression, we see that for all times $t$
\begin{equation}
L(q(t))\leq  \left( L(q(0))^\frac{1}{1+\alpha}-\epsilon (1+\alpha)^{-1} \gamma^{\alpha} t \right)_+^{1+\alpha}.
\end{equation}
The function $L(q)$ is continuous and therefore bounded above by a constant, $K$, for all values of $q$ with $||q||_1=1$. Hence, if $||q(0)||_1=1$, $L(q(t))=0$ for all $t\geq T$ where
\begin{equation*}
T=\frac{ (1+\alpha) K^{\frac{1}{1+\alpha}}}{\epsilon\gamma^{\alpha}},
\end{equation*}
and thus, as required, $q_j(t)=0$, $j\in\mJ$, for all $t\geq T$.
\end{proof}

The MaxWeight policies are designed to minimize the drift of a Lyapunov function. Notice, if we applied this rationale to $L$, above, then minimizing  \eqref{diffinequ0} would recover a form of MaxWeight-$\alpha$ policy. However,  an interesting feature of the approach of Bonald and Massoulie \cite{BoMa03} that is applied above is that the policy considered is not designed to minimize drift and thus resources are shared between different schedules rather than prioritized. This is an important feature when extending to multihop networks. Further, because of this feature, solutions to the optimization change continuously with changes in queue size. This could be important from a computation perspective when tracking the current optimal policy.

\subsection{Positive Recurrence}\label{sec8}
We are now in a position to combine Propositions \ref{FluidLimThrm} and \ref{FluidStabProof} to prove Theorem \ref{MainThrm}. We can at this point apply the general stability results of Dai \cite{Da95,Da95b} and Bramson \cite{Br08}.

\begin{proof}[Proof of Theorem \ref{MainThrm}]
For every $t\geq 0$, the sequence of queue sizes $\{ \bar{Q}^{(c)}(t)\}_{c\in\bN}$ is uniformly integrable. This is because the queue size process $\bar{Q}^{(c)}(t)$ is less than the number of arrivals $\bar{A}^{(c)}$, which is $L^2$ bounded, since we assume its increments have bounded variance. Now, by Proposition \ref{FluidLimThrm}, for any unbounded sequence in $\bN$, there is a subsequence $\{c_k\}_{k\in\bN}$ for which $\bar{Q}^{(c_k)}$ converges in distribution to fluid solution ${q}$. Let $T$ be the time given in Proposition \ref{FluidStabProof}, where $q_j(T)=0$ for $j\in\mJ$.
Since $\{|\bar{Q}^{(c_k)}(T)|\}_{c_k}$ is uniformly integrable and converges in distribution to ${q}(T)$, we also have $L_1$ convergence
\begin{equation}\label{fosters:fl1}
\lim_{c_k\rightarrow\infty} \bE ||\bar{Q}^{(c_k)}(T)||_1 = \bE ||q(T)||_1=0.
\end{equation}
This implies there exists a $\kappa$ such that for all $c> \kappa$
\begin{equation}\label{fosters:fl2}
\bE ||\bar{Q}^{(c)}(T)||_1  < (1-\epsilon).
\end{equation}
Note that if \eqref{fosters:fl2} did not hold then we could find a subsequence for which \eqref{fosters:fl1} did not hold; thus, we would have a contradiction. Expanding this inequality \eqref{fosters:fl2}, we have as described by Bramson \cite{Br08}, the following \emph{multiplicative Foster's condition}: for $\big|\big|{Q}(0)\big|\big|_1 > \kappa$
\begin{equation}\label{FosterCond}
\bE\left[ \big|\big|{Q}(T||{Q}(0)||_1)\big|\big|_1-\big|\big|{Q}(0)\big|\big|_1 \Big|  {Q}(0) \right] <  -\epsilon \big|\big|{Q}(0)\big|\big|_1 .
\end{equation}
This then implies our process is positive recurrent, see \cite[Proposition 4.6]{Br08}. 
\end{proof}

\section{Proof of Theorem 2}\label{sec6}
We give two proofs of Theorem \ref{MainThrm2}. The first proof reduces our multiclass fixed routing structure (Kelly type) to a single-class routing structure (Jackson type). From here we can apply the fluid stability proof of Massouli\'e \cite{Ma07}. The second proof observes the Lyapunov function arguments of Bramson \cite{Br96}, which apply for head-of-the-line processor sharing queueing networks, can be extended when services rates are proportionally fair instead of fixed.

\subsection{First Proof}

We state the following result on the fluid model for a single-class proportionally fair system with probabilistic routing. Massouli\'e \cite{Ma07} considered a process $(q_j(t) : j\in\mJ)$ satisfying, for $q_j(t)>0$, the following system of fluid equations
\begin{equation}\label{Fluid2}
\frac{d q_j}{dt} = a_j + \sum_{l\in\mJ} \sigma^*_l(q) p_{lj}   - \sigma^*_j(q),
\end{equation}
where $\sigma^*_j(q)$ gives the solution to the proportionally fair optimization, \eqref{MW-sq}.
The matrix $P=(p_{lj})_{lj}$ gives the probability that a link $l$ job next joins link $j$. So all packets depart this system, it is assumed that this matrix is sub-stochastic. From this one can calculate from external arrival rates $a=(a_j:j\in\mJ)$ the load induced on each link as
\begin{equation}
\bar{a}= (1- P^\T )^{-1} a.
\end{equation}
Massouli\'e proves the following fluid stability result.

\begin{theorem}[Massouli\'e '07]
If $\bar{a}$ belongs to the interior of the set $\coS$ then the fluid equations \eqref {Fluid2} are stable.
\end{theorem}

Since in Massouli\'e's model each packet leaving a link must have the same routing behavior, it is not immediately clear that the result extends to multiclass queueing networks where routing depends on a packet's route-class. However, the following lemma makes this possible
\begin{lemma}\label{LogLem}
For $j\in\mJ$, take positive vector $x^j=(x_{jr}: r\in \mR, j\in r)$ and positive constant $\sigma_j$. Let $q_j$ be the sum of the components of $x^j$. The following holds
\begin{align}
&q_j\log \sigma_j + \sum_{r: j\in r} x_{jr} \log \frac{x_{jr}}{q_j} \label{LogShare}\\
= &\max_{\gamma^j\geq 0} \; \sum_{r: j\in r} x_{jr} \log \gamma_{jr} \; \text{s.t.} \quad \sum_{j \in r} \gamma_{jr} = \sigma_j
\end{align}
and the above optimization is solved by $\gamma_{jr} = \frac{x_{jr}\sigma_j}{q_j}$.
\end{lemma}
The lemma is left as an exercise. But what the lemma says is that if packets $x^j=(x_{jr}: r\in \mR, j\in r)$ are sharing the resource $j$ then under a proportional fair optimization they share it proportionately and, perhaps more importantly, the equality \eqref{LogShare} shows that the form of the logarithm is the optimization remains unchanged after this optimization step. Thus, to solve a proportionally fair optimization given $x_{jr}$, we can solve the proportionally fair optimization given $q_j$, i.e. optimizing the $q_j\log \sigma_j$ terms in \eqref{LogShare}, and then assign rates according to a processor sharing discipline: $\gamma_{jr} = \frac{x_{jr}\sigma_j}{q_j}$. This describes the per-link queueing mechanism, PS2, described in Section \ref{MH-PF-sec}, which possesses significant structural simplification when compared to other policies known to be maximum stable for multi-hop routed switch networks.

\begin{proof}[Proof of Theorem \ref{MainThrm2}]
As a consequence of Lemma \ref{LogLem}, we have that the optimization
\begin{subequations}\label{PFopty1}
\begin{align}
&\text{maximize}  \quad \sum_{j\in\mJ} \sum_{r\in\mR} x_{jr} \log \gamma_{jr} \\
 &\text{subject to } \sum_{r: j \in r} \gamma_{jr} = \sigma_j, \; j\in\mJ,\\
&\text{over}\qquad\qquad \sigma \in \coS,
\end{align}
\end{subequations}
is equivalent to solving the optimization
\begin{subequations}\label{PFopty2}
\begin{align}
&\text{maximize}  \quad \sum_{j\in\mJ} q_{j} \log \sigma_{j} \\
&\text{over}\qquad\qquad \sigma \in \coS,
\end{align}
\end{subequations}
and setting $\gamma^*_{jr}(x)= x_{jr}\sigma^*_j(q)/q_j$ where $\sigma^*_j(q)$ solves the above proportionally fair optimization.

Consider process $x(t)=(x_{jr}(t): r\in\mR, j\in r)$, a Massouli\'e network of the type \eqref{Fluid2} where classes are now indexed by $(j,r)$. That is the stations of the network are no longer indexed by $j$, but by pairs class $(j,r)$. After completeing service at $(j,r)$ packets next join $(j^+_r,r)$ with probability one. This system obeys the fluid equations\footnote{Here $\gamma^{*}_{j^r_- r} (x(t))=\bar{a}_r$ when $j$ is the first queue on route $r$.}
\begin{align}
&\frac{d x_{jr}}{dt} =
\gamma^{*}_{j^r_- r} (x(t)) -  \gamma^{*}_{jr}(x(t)).
\end{align}
These equations are exactly of the form \eqref{Fluid2}. Thus, by Massouli\'e \cite{Ma07}, this system of fluid equations is stable. However, the equivalence of the optimizations \eqref{PFopty1} and \eqref{PFopty2} shows that this fluid system is the the same as our proportionally fair fluid model:
\begin{align*}
&\frac{d x_{jr}}{dt} =
\frac{x_{j^r_- r}(t)}{q_{j^r_-}(t)} \sigma^{*}_{j^r_-} (q(t)) - \frac{x_{j r}(t)}{q_{j}(t)} \sigma^{*}_{j}(q(t)).
\end{align*}
Thus the stability of our fluid model now follows from Massouli\'e \cite{Ma07}.
\end{proof}
The observation of Lemma \ref{LogLem} and its implications to multi-hop switch systems is critical.
After observing Lemma \ref{LogLem}, the reduction to Massouli\'e's proof is straightforward. However, to this author at least, it was surprising that the multi-class system can be reduced to a single-class system in this way. For instance, it is not immediately clear that one can reduce stationarity results on BCMP-Kelly networks as a direct reduction from results on Jackson networks. Nonetheless, the fluid stability proof of Massouli\'e is not so straightforward. A second, direct proof can argued from Bramson \cite{Br96}. As a brief aside: with Lemma \ref{LogLem} and Lemma \ref{PartialH} (below),  the fluid stability result of Bramson for processor-sharing networks \cite{Br96} and the fluid result of Massouli\'e for bandwidth networks \cite{Ma07} have some equivalence. However, such equivalences are, perhaps, less striking than their resulting implications for switch systems. We now see that, for the proportionally fair case, stability of a single-hop queueing network also implies stability of its multi-hop routed counterpart. Moreover, the queueing structure of this system is far simpler than BackPressure, the canonical class of multi-hop throughput optimal switch policies.

\section{Second Proof}

Bramson considers the fluid model of a multiclass processor sharing queueing network. Here the service rate at queues $\sigma\in\bR_+^\mJ$ is not a function of queue sizes. So, for $x_{jr}(t)>0$, the fluid model of Bramson is as follows
\begin{align*}
&\frac{d x_{jr}}{dt} = \frac{x_{j^r_- r}(t)}{q_{j^r_-}(t)} \sigma_{j^r_-}  - \frac{x_{j r}(t)}{q_{j}(t)} \sigma_{j}.
\end{align*}
Bramson considers a Lyapunov function
\begin{equation}
\hat{H}(x) =  \sum_{r\in\mR} \sum_{j\in r} x_{jr} \log \left( \frac{x_{jr}\sigma_j}{q_j\bar{a}_r} \right).
\end{equation}
This is large deviations rate function associated with a network of processor sharing queues.
Informally speaking, his proof relies on the fact that the partial derivatives of this function are, for $x_{jr}>0$,
\begin{equation}
 \frac{\partial \hat{H}}{\partial x_{jr}} =\log \left( \frac{x_{jr}\sigma_j}{q_j\bar{a}_r} \right).
\end{equation}
 If we now replace $\sigma$ with the proportionally fair solution $\sigma^*(q)$ and consider the following Lyapunov function
\begin{equation}
H(x) =  \sum_{r\in\mR} \sum_{j\in r} x_{jr} \log \left( \frac{x_{jr}\sigma^*_j(q)}{q_j\bar{a}_r} \right),
\end{equation}
due to the properties of the proportionally fair optimization, the partial derivatives of this optimization remain the same. This is demonstrated in the following lemma.
\begin{lemma}\label{PartialH} For $x_{jr}>0$
\begin{equation}
\frac{\partial H}{\partial x_{jr}} =  \log \left( \frac{x_{jr}\sigma^*_j(q)}{q_j\bar{a}_r} \right).
\end{equation}
\end{lemma}

\begin{proof}[Proof of Lemma \ref{PartialH}]
We take the partial derivative of $H(x)$ with respect to $x_{jr}$ from first principles.
For $h>0$, define $x^h_{jr}= x_{jr} + h$ and $x^h_{j'r'}= x_{j'r'}$ for $r'\neq r$ and $j'\neq j$. Then
\begin{align*}
&\frac{H(x^h)- H(x)}{h}\\
 =& \frac{1}{h}
\bigg[
\sum_{r\in\mR} \sum_{j\in r} x^h_{jr} \log \left( \frac{x^h_{jr}\sigma^*_j(q^h)}{q^h_j\bar{a}_r} \right)\\
&\qquad\qquad -
\sum_{r\in\mR} \sum_{j\in r} x_{jr} \log \left( \frac{x_{jr}\sigma^*_j(q)}{q_j\bar{a}_r} \right)
  \bigg]\\
\geq &
 \frac{1}{h}
\bigg[
\sum_{r\in\mR} \sum_{j\in r}x^h_{jr} \log \left( \frac{x^h_{jr}\sigma^*_j(q)}{q^h_j\bar{a}_r} \right)\\
&\qquad\qquad -\sum_{r\in\mR} \sum_{j\in r} x_{jr} \log \left( \frac{x_{jr}\sigma^*_j(q)}{q_j\bar{a}_r} \right)
  \bigg]\\
=& \log\left(\frac{\sigma^*_{j}(q)}{\bar{a}_r}\right)  +  \frac{1}{h} \left( x^h_{jr} \log x^h_{jr} - x_{jr} \log x_{jr} \right) \\
&\quad + \frac{1}{h} \left( q^h_{j} \log q^h_{j} - q_{j} \log q_{j} \right).
\end{align*}
In the inequality above, we use the fact ${\sigma}^*_{j}(q)$ is suboptimal for proportional fair problem with parameter choice $q^h$.
Applying the same argument but instead observing that $\sigma^*_{j}(q^h)$ is suboptimal
for parameter choice $q$, we can also derive the bound
\begin{align}
&\frac{H(x^h)- H(x)}{h} \notag\\
 \leq & \log\left(\frac{\sigma^*_{j}(q^h)}{\bar{a}_r}\right)  +  \frac{1}{h} \left( x^h_{jr} \log x^h_{jr} - x_{jr} \log x_{jr} \right) \\
&\quad + \frac{1}{h} \left( q^h_{j} \log q^h_{j} - q_{j} \log q_{j} \right). \label{PFDIff}
\end{align}
We remark that these same two bounds can also be derived in the case $h<0$.
The function $q \mapsto \sigma_{j}(q)$ is a continuous function for $q_j>0$. This is proven in \cite[Lemma A.3]{KeWi04}.
Thus taking the limit as $h\rightarrow 0$, both of these bounds together imply
\begin{align*}
\frac{\partial H(x)}{\partial x_{jr}} &= \log\left(\frac{\sigma^*_{j}(q)}{\bar{a}_r}\right) + \log x_{jr} - \log q_j,
\end{align*}
as required.
\end{proof}

With the addition of a few technical lemmas, we can now give a direct proof of Theorem \ref{MainThrm2}. These additional technical lemmas are stated and proven after the proof of our main result.

\begin{proof}[Proof of Theorem \ref{MainThrm2}]
It is proven in Lemma \ref{PosH} that $H(x)$ is positive and is minimized when $x=0$. For $x_{jr}>0$,
 it is proven in Lemma \ref{PartialH} that the partial derivatives of $H(x)$ with respect to $x_{jr}$ are
\begin{equation}\label{HHPartial}
\frac{\partial H}{\partial x_{jr}} =  \log \left( \frac{x_{jr}\sigma^*_j(q)}{q_j\bar{a}_r} \right).
\end{equation}
Thus we can see that the following equalities hold
\begin{align}
&\frac{d H}{dt} = \sum_{r, j\in r} \frac{d x_{jr}}{dt} \frac{\partial H}{\partial x_{jr}}\\
=& \sum_{r, j\in r} \left( \frac{x_{j^r_- r}}{q_{j^r_-}} \sigma^{*}_{j_-^r} (q) - \frac{x_{j r}}{q_{j}} \sigma^{*}_{j}(q) \right) \log \left( \frac{x_{jr}\sigma^*_j(q)}{q_j\bar{a}_r} \right) \label{Hdiffy0} \\
=&
-
\sum_{r\in\mR} \bar{a}_r \sum_{j\in r}
	 \frac{x_{j r}\sigma^{*}_{j}(q) }{q_{j}\bar{a}_r}
		\log
			\left(
				\frac{
					\big[\frac{x_{jr}\sigma^*_j(q)}{q_j\bar{a}_r}\big]
				}{
					\Big[\frac{x_{j^r_+ r}\sigma^*_{j^r_+}(q)}{q_{j^r_+}\bar{a}_r}\Big]
				}
			\right). \label{Hdiffy}
\end{align}
The first we apply chain rule. For \eqref{Hdiffy0}, we substitute \eqref{HHPartial} and the definition of our fluid model \eqref{hopfluid}.
To derive \eqref{Hdiffy}, we increment the first terms of summation \eqref{Hdiffy0} so that the coefficients of the logarithms are equal. After multiplying and dividing by $\bar{a}_r$ we gain expression \eqref{Hdiffy}.  Notice in the first equality, we include the term $j=j_0^r$. Since applying convention \eqref{Convension} we have, $ \log ( (x_{jr}\sigma^*_j(q_j))/(q_j\bar{a}_r) )=\log 1 =0$.

The following sequence of inequalities hold, for $q\neq 0$
\begin{align*}
\frac{d H}{dt} \leq  - \sum_{r\in\mR} \Bigg[ \bar{a}_r &  \bigg(\sum_{j\in r} \frac{x_{jr}\sigma^*_{j}(q)}{q_{j}\bar{a}_r}\bigg)^{-1} \\
& \times \sum_{j\in r} \bigg( \frac{x_{j r}\sigma^*_{j^r}(q)}{q_{j^r}\bar{a}_r}-  \frac{x_{j^r_+ r}\sigma^*_{j^r_+}(q)}{q_{j^r_+}\bar{a}_r}  \bigg)^2 \Bigg]\\
\leq - \sum_{r\in\mR} \Bigg[ \bar{a}_r &  \bigg(\frac{\sigma_{max}}{\bar{a}_r} \bigg)^{-1} \\
& \times \sum_{j\in r} \bigg( \frac{x_{jr}\sigma^*_{j}(q)}{q_{j}\bar{a}_r}-  \frac{x_{j^r_+ r}\sigma^*_{j^r_+}(q)}{q_{j^r_+}\bar{a}_r}  \bigg)^2 \Bigg]\\
 \leq -\epsilon. \qquad\quad&
\end{align*}
In the first inequality, we apply entropy bound Lemma \ref{pinklem} to \eqref{Hdiffy}. In the second inequality,
we apply the bound
\begin{equation}
\sum_{j\in r} \frac{x_{jr}\sigma^*_{j}(q)}{q_{j}\bar{a}_r} \leq \frac{\sigma_{max}}{\bar{a}_r}.
\end{equation}
In the third inequality, we apply the Lemma \ref{LyuEpsLem}. This then shows that our Lyapunov function has strictly negative drift.

By Lemma \ref{PosH}, $H(x)=0$ only if $x=0$ and $H(x)$ is bounded when $||x||_1=1$ . Thus, for all $x(0)$ such that $H(x(0)) \leq h$ for some constant $h>0$ we have for all $t \geq h\epsilon^{-1}$
\begin{equation}
q(t) = 0.
\end{equation}
Thus, fluid stability holds.
\end{proof}

We now prove the lemmas required in the above proof.
For two probability distributions $p$ and $q$ defined on the same finite set $\mX$, the relative entropy between  $p$ and $q$  is defined to be
\begin{equation*}
D(p || q) = \sum_{x\in\mX} p_x \log \frac{p_x}{q_x}.
\end{equation*}
The following bound on $D(p||q)$ holds

\begin{lemma}[Pinsker's Inequality]
\begin{equation*}
 \sqrt{D(p||q)} \geq \sum_{x\in\mX} | p_x - q_x|.
\end{equation*}
\end{lemma}
A proof of this bound is found in Cover and Thomas \cite{Co91}. Relative entropy is positive continuous in $p$ and minimized when $p=q$. This gives the following.

\begin{lemma}\label{PosH}
The function $H(x)$ is positive, is bounded when $||x||_1=1$ and is minimized when $x=0$.
\end{lemma}
\begin{proof}
Observe that, $H(x)$ can be expressed as linear combination of relative entropy terms as follows
\begin{align}
H(x)= &\sum_{j\in\mJ} q_j D\left( \Big(\frac{x_{jr}}{q_j}\Big)_{r \ni j} \bigg| \bigg|\Big(\frac{\bar{a}_r}{a_j}\Big)_{r \ni j} \right) + \sum_{j\in\mJ} q_j \log \frac{\sigma_j^*(q)}{a_j} \notag \\
\geq &\sum_{j\in\mJ} q_j \log \frac{\sigma_j^*(q)}{a_j} \geq 0.
\end{align}
The first inequality follows by the positivity of the relative entropy. The second inequality follows by the optimality of $\sigma_j^*(q)$ and the fact the vector $(a_j: j\in\mJ)$ belongs to the interior of $\coS$. From the form of the entropy equality above, we see that $H(x)$ is continuous for $||x||_1=1$ and so is bounded. Further, we note that the inequalities above hold with equality iff $x = 0$.
\end{proof}

If we now do not assume that $p$ and $q$ are probability distributions but instead we assume that they are positive and sum to the same constant, then $D(p || q)$ is still well defined
the following is a consequence of Pinsker's Inequality.
\begin{lemma}\label{pinklem}
If $p$ and $q$ are two positive vectors components indexed by $\mX$ and with
\begin{equation}
\sum_{x\in\mX} p_x = \sum_{x\in\mX} q_x
\end{equation}
then
\begin{equation}
\sum_{x\in\mX} p_x \log \frac{p_x}{q_x} \geq \frac{1}{\sum_{x\in\mX} p_x}\cdot \sum_{x\in\mX} \left( p_x - q_x \right)^2.
\end{equation}
\end{lemma}
\begin{proof}
Define
\begin{equation}
\tilde{p}_x=\frac{p_x}{\sum_{y\in\mX} p_y} \quad \text{and}\quad \tilde{q}_x=\frac{q_x}{\sum_{y\in\mX} q_y}
\end{equation}
Then
\begin{align*}
\sum_{x\in\mX} p_x \log \frac{p_x}{q_x} &=D(\tilde{p} || \tilde{q}){\sum_{x\in\mX} p_x}  \\
\geq &  \left(  \sum_{x\in\mX} | \tilde{p}_x - \tilde{q}_x | \right)^2  {\sum_{x\in\mX} p_x}\\
= & \frac{1}{\sum_{x\in\mX} p_x}\left(  \sum_{x\in\mX} | {p}_x - {q}_x | \right)^2\\
\geq &\frac{1}{\sum_{x\in\mX} p_x}  \sum_{x\in\mX} ({p}_x - {q}_x )^2.
\end{align*}
In the first equality above, we apply the definition of $D(\tilde{p}||\tilde{q})$; in the second, we apply Pinsker\rq{}s Inequality; we then rearrange this term and bound to get the result.
\end{proof}

The follow lemma is used to prove that $H(t)$ has strictly negative drift.

\begin{lemma} \label{LyuEpsLem} There exists $\epsilon >0$ such that for $x\in\bR_+^{\mR \times \mJ}$
\begin{equation}\label{LyaEps}
\sum_{r\in\mR}\bar{a}_r\bigg(\frac{\sigma_{max}}{\bar{a}_r} \bigg)^{-1} \sum_{j\in r} \bigg( \frac{x_{jr}\sigma^*_{j}(q)}{q_{j}\bar{a}_r}-  \frac{x_{j^r_+r}\sigma^*_{j^r_+}(q)}{q_{j^r_+}\bar{a}_r}  \bigg)^2 > \epsilon.
\end{equation}
\end{lemma}
\begin{proof}
First let\rq{}s look what is required for the term on  the left-hand side of \eqref{LyaEps} to equal zero -- which, by the statement of the lemma, should give a contradiction. We would require that for each route $r\in\mR$, for all $j\in r$
\begin{equation}
 \frac{x_{jr}\sigma^*_{j}(q)}{q_{j}\bar{a}_r}-  \frac{x_{j^r_+ r}\sigma^*_{j^r_+}(q)}{q_{j^r_+}\bar{a}_r}=0.
\end{equation}
So the terms  $\frac{x_{jr}\sigma^*_{j}(q)}{q_{j}\bar{a}_r}$ are constant over $j\in r$, including $j_0$. But $\frac{x_{j^r_0r}\sigma^*_{j^r_0}(q)}{q_{j^r_0}\bar{a}_r}=1$. So we would require that, for all $r$ and $j\in r$
\begin{equation*}
\bar{a}_r = \frac{x_{jr}\sigma^*_{j}(q)}{q_{j}}.
\end{equation*}
Now summing over $r\in j$ we would require that
\begin{equation}
\bar{a}_j = \sum_{r: j\in r} \bar{a}_r = \sigma^*_{j}(q).
\end{equation}
However, we now have a contradiction: by assumption $(\bar{a}_j : j\in\mJ)$ does not belong to the boundary of the capacity set $\coS$ while $(\sigma^*_j(q): j\in\mJ)$ does (because it is Pareto optimal). So the above inequality cannot hold. So the  left-hand side of \eqref{LyaEps} cannot equal zero.

Now let\rq{}s use this last contradiction to extrapolate back and prove the lemma. Since $(\bar{a}_j:j\in\mJ)$ belongs to the interior of $\coS$ and $\sigma^*_j(q)\in \partial\!\! \coS$, there exist an $\delta>0$ such that for each vector $q\neq 0$ there is some $j\in\mJ$ such that
\begin{equation}
\bar{a}_j + \delta < \sigma^*_{j}(q).
\end{equation}
Thus, for the above inequality to hold, there exists some $r$ with $j\in r$ such that
\begin{equation}
\bar{a}_r + \frac{\delta}{|r|} < \frac{x_{jr}}{q_j}\sigma^*_{j}(q)
\end{equation}
or expressed differently,
\begin{equation}
\frac{x_{j^r_0r}\sigma^*_{j^r_0}(q)}{q_{j^r_0}\bar{a}_r} + \frac{\delta}{|r|\bar{a}_r} < \frac{x_{jr}\sigma^*_{j}(q)}{q_{j}\bar{a}_r}.
\end{equation}
This for this to hold between $j_0^r$ and $j$ on route $r$ there must be two links in the sequence $l$ and $l'$ on route $r$ between  $j_0^r$ and $j$ such that
\begin{equation}
 \frac{\delta}{|r|^2a_r} < \frac{x_{lr}\sigma^*_{l}(q)}{q_{l}\bar{a}_r} - \frac{x_{l'r}\sigma^*_{l'}(q)}{q_{l'}\bar{a}_r}.
\end{equation}
Setting $\epsilon = ( \frac{\delta}{|r|^2})^2 \sigma_{\max}$, we see that \eqref{LyaEps} must hold.
\end{proof}

\section{Conclusions and Future Work}
In this paper,  we have considered a generalized class of switch policies. We have shown that these share the same maximum stability properties as MaxWeight of policies. For the proportionally fair class of schedulers, we have extended these stability results to multihop networks with fixed routing. This is the first general proof of maximum stability for a multihop switch network where routing information is not required in order to make a scheduling decision. This is also significant given the greatly simplified queueing structure required to define the policy.

Future work might consider the implementation of our policies, which require the solution of a concave optimization problem. Thus, using for instance an interior point method, one can approximate a solution to this optimization in computational time that is polynomial in the number of constraints of the set of feasible schedules $\coS$. After this one can decompose the mean vector onto the set of schedules $\mS$. For instance, in the case of a Input Queued switch one could use a Birkhof von Neumann decomposition. Approaches of this type have been considered previously \cite{ShWaZh12}. Compared to MaxWeight, an advantage of the strict convexity of our optimization is that the solutions of our optimization move continuously with continuous changes in queue size. Thus one could consider an alternative, online convex optimization schema where one continuously adapts the current schedule in order to track the optimum $(\alpha,g)$-policy. As queue sizes increase the solution of our optimization will change more slowly and thus a tracking policy would expect to converge to the correct scheduling decision \cite{zinkevich2003online,hazan2007logarithmic}.

We have emphasized the structural benefits of implementing a Proportional Scheduler but have not thus far discussed the statistical benefits of the policy. As we mentioned in the introduction, there is a close relationship between proportional fairness and the maximum stability of reversible systems. With this, one can begin to show that the randomized queueing discipline considered in the Proportional Scheduler can be replaced by a FIFO service discipline and still maintain its stability properties. This is important from a practical perspective as most communications systems implement FIFO queueing. From a theoretical perspective, the result is far more involved and technical than the relatively straightforward proofs for the randomized service which fit within this paper. However, with this analysis further statistical benefits can be observed. For instance, delay is known to increase quadratically with route length for the BackPressure policy, see  \cite{BSS11,St11}. However, one can argue that for the Proportional Scheduler delay will grow linearly with route length.

As mentioned above, we have given a proof where a single-hop maximum stability policy implies stability of its routed multihop counterpart, and, crucially, these scheduling decision are made without discriminating between packet route-classes. It is reasonable to conjecture that this holds more generally than for just the Proportional Scheduler. The hope is that the arguments started here could give a general proof technique. One could then hope to extend recent, notable developments on single-hop networks \cite{MMAW99,St04,JiWa10,ShSh12} to multihop networks. Thus we begin to provide performance analysis of switched communication networks where packets are {communicated}.

\bibliographystyle{amsplain}
\bibliography{MW-AB}

\end{document}